\documentclass[9pt,twocolumn,twoside]{pnas-new}

\templatetype{pnasresearcharticle} 

\usepackage{svg}
\usepackage{bm}
\usepackage{mathtools}
\usepackage{amsmath}

\algrenewcommand\algorithmicrequire{\textbf{Input:}}
\algrenewcommand\algorithmicensure{\textbf{Output:}}

\algrenewcommand\algorithmicrequire{\textbf{Input:}}
\algrenewcommand\algorithmicensure{\textbf{Output:}}
\usepackage{amsthm}
\usepackage{amsmath}
\newtheorem{theorem}{Theorem}[subsection]

\theoremstyle{remark}

\newtheorem{definition}{Definition}[subsection]

\usepackage{mathtools}
\usepackage{amssymb}

\usepackage {mathtools}
\usepackage{amsmath}
\usepackage{algpseudocode}
\usepackage{algorithm}
\usepackage{CJKutf8}

\algnewcommand\algorithmicforeach{\textbf{for each}}
\algdef{S}[FOR]{ForEach}[1]{\algorithmicforeach\ #1\ \algorithmicdo}

\title{Robust estimations from distribution structures: V. Non-asymptotic
}

\author[a,b,c,1]{Li Tuobang}

\leadauthor{Li}

\significancestatement{
Most contemporary statistics theories focus on asymptotic analysis due to its tractability and simplicity. Non-asymptotic statistics are crucial when dealing with small or moderate sample sizes, which is often the case in practice. In situations where analytical results are difficult or impossible to obtain, Monte Carlo studies serve as a powerful tool for addressing non-asymptotic behavior. However, these studies can be computationally expensive, particularly when high precision is required or when the statistical model demands significant computational time. Here, we propose calibrated Monte Carlo study that aims to approximate the randomness structures using a small set of sequences. This approach sheds light on understanding the general structure of randomness.

}
\correspondingauthor{\textsuperscript{1}To whom correspondence should be addressed. E-mail: tuobang@biomathematics.org}

\keywords{finite sample bias $|$ order statistics $|$ variance reduction $|$ Monte Carlo study $|$ uniform distribution} 

\begin{abstract}
Due to the complexity of order statistics, the finite sample behaviour of robust statistics is generally not analytically solvable. While the Monte Carlo method can provide approximate solutions, its convergence rate is typically very slow, making the computational cost to achieve the desired accuracy unaffordable for ordinary users. In this paper, we propose an approach analogous to the Fourier transformation to decompose the finite sample structure of the uniform distribution. By obtaining sets of sequences that are consistent with parametric distributions for the first four sample moments, we can approximate the finite sample behavior of other estimators with significantly reduced computational costs. This article reveals the underlying structure of randomness and presents a novel approach to integrate multiple assumptions. 
\end{abstract}

\dates{This manuscript was compiled on \today}

\begin{document}

\maketitle
\thispagestyle{firststyle}
\ifthenelse{\boolean{shortarticle}}{\ifthenelse{\boolean{singlecolumn}}{\abscontentformatted}{\abscontent}}{}


\dropcap{I}n the early nineteenth century, Bessel deduced the unbiased sample variance and found it has a correction term of $\frac{n}{n-1}$. One century later, Cramér \cite{cramer1999mathematical} in his classic textbook \emph{Mathematical Methods of Statistics} deduced unbiased sample central moments with a linear time complexity \cite{gerlovina2019computer}. However, apart from the mean and central moments, the finite sample behavior of nearly all other estimators depends on the underlying distribution and lacks a simple non-parametric correction term. A widely used exact finite sample bias correction for unbiased standard deviation in the Gaussian distribution is the factor, $\sqrt{\frac{n-1}{2}}\frac{\Gamma\left(\frac{n-1}{2}\right)}{\Gamma\left(\frac{n}{2}\right)}$, which can be deduced using Cochran's theorem \cite{holtzman1950unbiased}. While this correction is already more complex than Bessel's correction, for robust statistics, the situation is even much more complicated. For example, the simplest robust estimator, the median, exhibits a highly complex finite sample behavior. If $n$ is odd, $E\left[median_n\right]=\int_{-\infty}^{\infty}{\left(\frac{n+1}{2}\right)\binom{n}{\frac{n}{2}-\frac{1}{2}}F\left(x\right)^{\frac{n}{2}-\frac{1}{2}}\left[1-F(x)\right]^{\frac{n}{2}-\frac{1}{2}}f(x)xdx}$ \cite{David2003OrderST}, where $F(x)$ and $f(x)$ represent the cumulative distribution function (cdf) and probability density function (pdf) of the assumed distribution, respectively. For the exponential distribution, the above equation is analytically solvable, yielding $E\left[median_n\right]=\frac{2^{-n-1} (n+1) \binom{n}{\frac{n-1}{2}} \left(H_n-H_{\frac{n-1}{2}}\right) \Gamma \left(\frac{n+1}{2}\right)\sqrt{\pi}}{\lambda  \Gamma \left(\frac{n}{2}+1\right)}$, where $H_n$ denotes the $n$th Harmonic number, $\Gamma$ represents the gamma function, and $\lambda$ stands for the scale parameter of the exponential distribution. However, for distributions with more complex pdfs, such equations are generally unsolvable. For more complex estimators, writing their exact finite-sample distribution formulas becomes challenging. In 2013, Nagatsuka, Kawakami, Kamakura, and Yamamoto derived the exact finite-sample distribution of the median absolute deviation, which consists of four cases, each with a lengthy formula \cite{NAGATSUKA2013999}. In such cases, even obtaining a numerical solution is challenging \cite{David2003OrderST,NAGATSUKA2013999}. So, Monte Carlo simulation is currently the only practical choice for estimating finite sample corrections. However, the computational cost of Monte Carlo simulation is often too high to be processed on a typical PC. For example, for median absolute deviation, Croux and Rousseeuw (1992) provided correction factors with a precision of three decimal places for $n\le 9$ using 200,000 pseudorandom Gaussian sample \cite{Croux1992TimeEfficientAF}. Hayes (2014) reported correction factors for $n\le 100$ using 1 million pseudorandom samples for each value of $n$ to ensure the accuracy to four decimal places \cite{Kevin3}. Recently, Akinshin (2022) \cite{akinshin2022finitesample} presented correction factors for $n\le 3000$ using 0.2-1 billion pseudorandom Gaussian samples. His result suggest that, for the median absolute deviation, finite sample bias correction is required to ensure a precision of three decimal places when the sample size is smaller than 2000. This highlights the importance of finite sample bias correction. However, since different correction factors are required for different parametric assumptions, the computational cost of addressing all possible cases in the real world becomes significant, especially for complex models. 

In addition to computational challenges, there exists an inherent difficulty in dealing with randomness. The theory of probability provides a framework for modeling and understanding random phenomena. However, the practical implementation of these models can be challenging, as discussed, and their complexity greatly hinders our comprehension. The quality of randomness can significantly impact the validity of simulation results, and a deeper understanding of randomness may offer a more effective and cost-efficient solution. The purpose of this article is to demonstrate that the finite sample structure of uniform random variables can be decomposed using a few well-designed sequences with high accuracy. Furthermore, we show that the computational cost of Monte Carlo study can be dramatically improved by obtaining sets of sequences that are consistent with multiple parametric distributions respectively for sample central moments.

\begin{figure*}
\centering
\includegraphics[width=1\linewidth]{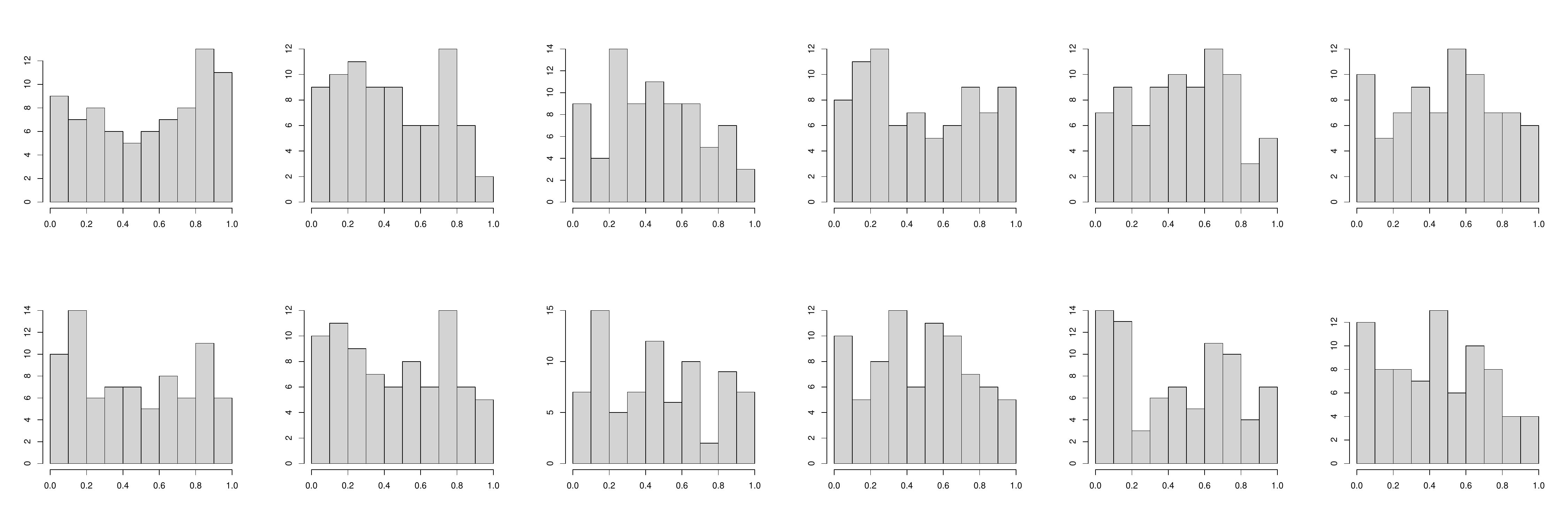}
\caption{The frequency histograms of pseudo-random sequences on the interval [0,1] with size 80.
}

\label{fig:Biasplot}
\end{figure*}

\section*{Decomposing the Finite Sample Structure of Uniform Distribution}\label{BB}
Any continuous distribution can be linked to the uniform distribution on the interval [0, 1] through its quantile function. This fundamental concept in Monte Carlo study implies that understanding the finite sample structure of uniform random variables can be leveraged to understand the finite sample structure of any other continuous random variable through the quantile transform. The Glivenko–Cantelli theorem \cite{glivenko1933sulla,cantelli1933sulla} ensures the almost-sure convergence of the empirical distribution function to the true distribution function. However, the individual empirical distribution often deviates significantly from the asymptotic distribution even when the sample size is not small (Figure \ref{fig:Biasplot}, sample size is 80), which cause finite sample biases and variance of common estimators. Let $\mu$, $\mu_2$, $\ldots$, $\mu_{\mathbf{k}}$ denote the first $\mathbf{k}$ central moments of a probability distribution. According to the unbiased sample central moment \cite{cramer1999mathematical}, the expected value of the sample central moment, $m_{\mathbf{k}}=\frac{1}{n}\sum_{{\mathbf{k}}=1}^{n}\left(x_{\mathbf{k}}-\bar{x}\right)^{\mathbf{k}}$, can be deduced, denoted as $E[m_{\mathbf{k}}]$. Let $S = \{sequence[i] | i \in \mathbb{N}\}$ be a set of number sequences ranging from 0 to 1, where $sequence[i]$ represents the $i$th sequence in the set, and $\mathbb{N}$ is the set of natural numbers, with $i\le N$. Transform every number in $S$ using the quantile function of a parametric distribution, $PD$. The transformed sequences can be denoted as $S_{PD}$. Denote the set of the $\mathbf{k}$th sample central moments for these transformed sequences as $M_{\mathbf{k}}=\{m_{\mathbf{k},i}| i \in \mathbb{N}\}$. $S$ is consistent with $PD$ for all $m_{\mathbf{k}}$ when $\mathbf{k}\le k$, if and only if the following system of linear equations is consistent, $\begin{cases} m_{1,1}w_1+\ldots+m_{1,i} w_i+\ldots+m_{1,N} w_N=E[m_{\mathbf{1}}]  \\ \ldots\\ m_{{\mathbf{k}},1} w_1+\ldots+m_{{\mathbf{k}},i} w_i+\ldots+ m_{{\mathbf{k}},N} w_N=E[m_{\mathbf{{\mathbf{k}}}}] \\ \ldots\\ m_{k,1} w_1+\ldots+m_{k,i} w_i+\ldots+m_{k,N} w_N=E[m_{k}]\\
w_1+\ldots+w_i+\ldots+w_N=1
\end{cases}$, where $w_1$, $\ldots$, $w_i$, $\ldots$, $w_N$ are the unknowns of the system, with $N\geq k+1$. $w_1$, $\ldots$, $w_i$, $\ldots$, $w_N$ can be determined using a typical constraint optimization algorithm. The Monte Carlo study can be seen as a special case when $w_1=\ldots=w_i=\ldots=w_N$, and the sequences in $S$ are all random number sequences. The strong law of large numbers (proven by Kolmogorov in 1933) \cite{kolmogorov1933sulla} ensures that in this case, when the number of sequences $N\to\infty$ or when the sample size $n\to\infty$, the above system of linear equations is always consistent. Another trivial but important condition is introduced in the following theorem.
\begin{theorem}\label{whlk} For any set of sequences, there is always a probability distribution that it is consistent with.\end{theorem}\begin{proof} Since a sequence can be seen as a discreet probability distribution, the conclusion is trivial if assigning one weight as 1 and other weights as zeros. \end{proof}

\begin{figure*}
\centering
\includegraphics[width=1\linewidth]{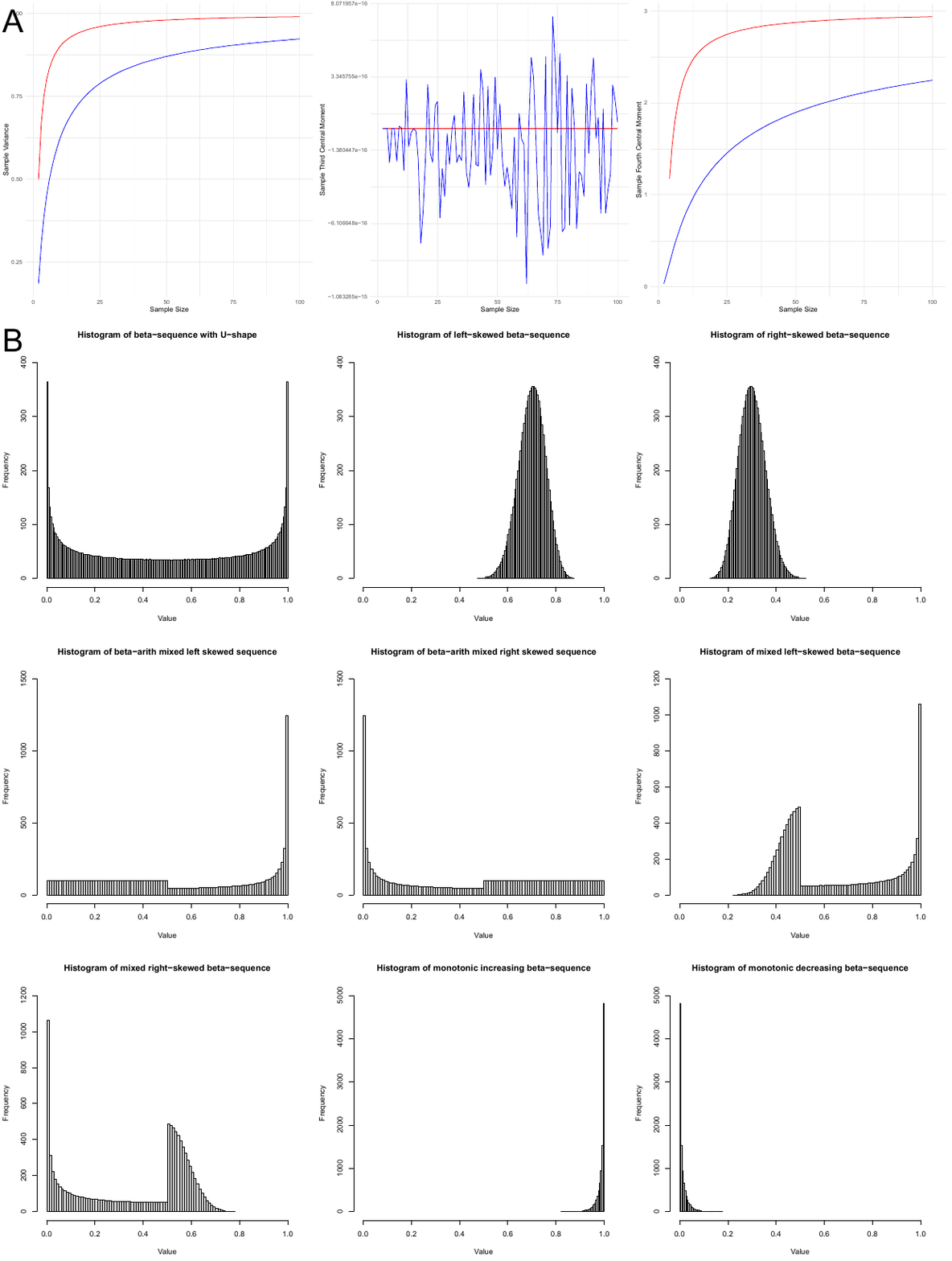}
\caption{A. The first four sample central moments for the Gaussian distribution are plotted over a sample size ranging from 5 to 100. The red lines represent the expected values, while the blue lines depict the values estimated from the arithmetic sequences. B. The histograms of different beta-sequences, their self-mixtures, and mixtures with arithmetic sequences.
}

\label{fig:2}
\end{figure*}

Low-discrepancy sequences are commonly used as a replacement of uniformly distributed random numbers to reduce computational cost. When considering a sequence to approximate the structure of uniform random variables, the most natural choice is the arithmetic sequence, denoted as $\{ x_i \}_{i=1}^{n} = \left\{ \frac{i}{n+1} \right\}_{i=1}^{n}$. As $n\to\infty$, the arithmetic sequence towards the continuous uniform distribution. However, the arithmetic central moments estimated from the arithmetic sequence for the Gaussian distribution differ significantly from their expected values (Figure \ref{fig:2}A). The arithmetic sequence lacks the variability of true random samples which produce additional biases for even order moments. The beta distribution is defined on the interval (0, 1) in terms of two shape parameters, denoted by $\alpha$ and $\beta$. When $\alpha=\beta$, the beta distribution is symmetric. To better replicate the features of uniform random variables, we introduced beta distributions with a variety of parameters. The arithmetic sequences were transformed by the quantile functions of these beta distributions to form beta-sequences, resulting in sequences that are U-shape ($\alpha=\beta=0.547$), left-skewed ($\alpha=46.761$, $\beta=20.108$), right-skewed ($\alpha=20.108$, $\beta=46.761$), monotonic decreasing ($\alpha=0.478$, $\beta=38.53$), monotonic increasing ($\alpha=38.53$, $\beta=0.478$), their left-skewed self-mixtures ($\alpha=\beta=0.369$, $\alpha=\beta=18.933$), their right-skewed self-mixtures ($\alpha=\beta=0.369$, $\alpha=\beta=18.933$), their left-skewed mixture with the arithmetic sequence ($\alpha=\beta=0.328$), their right-skewed mixture with the arithmetic sequence ($\alpha=\beta=0.328$) (Figure \ref{fig:2}B). Besides beta-sequences with a U-shape, other sequences are paired by symmetry so additional constraints is set to ensure equal weight for each pair. Besides these 9 sequences and arithmetic sequences, a pseudo-random sequence is introduced to further approximate the structure and avoid inconsistent scenarios. Finally, a complement sequence is introduced which if combining all the sequences with corresponding weights, the overall sequence is nearly uniform.

\begin{figure*}
\centering
\includegraphics[width=1\linewidth]{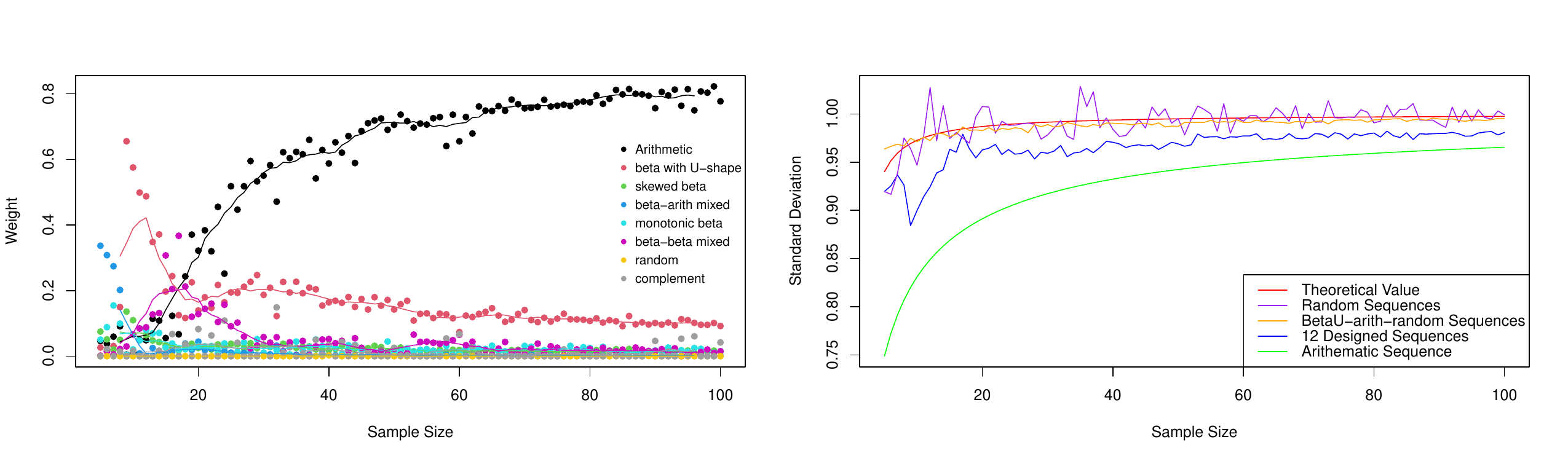}
\caption{The first plot shows the weights assigned to different sequences as the sample size increases. The second plot depicts the sample standard deviations estimated from random, beta-arith-random, 12 designed, and arithmetic sequences and compares them to the true values. The numbers of random, betaU-arith-random, and 12 designed sequences were both 120 in total.}

\label{fig:3}
\end{figure*}
\section*{Estimating Finite Sample Bias from Calibrated Sequences}\label{AA}
The most surprising result in this article is that, by carefully selecting/designing sequences in $S$, even when $N$ and $n$ are very small, e.g., less than 20, the above system of linear equations can still be consistent, while the weight assigns to the random and complement sequences are extremely small (<0.02 on average). Using just 12 sequences, when $n=10$, the constraint optimization algorithm can assign weights to all these sequences with errors less than $10^{-10}$, when $k\le4$. This means that technically, these sequences are consistent with the Gaussian distribution for the first four moments. More importantly, in the case of Gaussian distirbution, the findings suggest the arithmetic sequence and beta sequence with a U-shape accounts for more than 70\% of the finite sample properties of uniform random variables, while monotonic beta, beta-beta mixed, skewed beta distributions each contribute about 2-10\% (Figure \ref{fig:3}). As the sample size grows, as expected, the weight of the arithmetic sequence increases and dominants while the weights of other sequences gradually decrease. However, the beta-sequence with a U-shape still holds about 10\% weight even when the sample size is 100 (Figure \ref{fig:3}).

The obtained weights can be used to estimate the finite sample behaviour of other related estimators, such as the standard deviation for the Gaussian distribution. We found that by using the 12 well-designed sequences, the estimation of finite sample bias is much better than the arithmetic sequence (Figure \ref{fig:3}), but the accuracy is still poor. The RMSE is 0.0274 after repeated 10 times to reduce the variance caused by the random sequences, higher than that of the Monte Carlo study with the same size of sequences, which is 0.0127. To further increase accuracy, we adopted a stochastic method. Since the arithmetic and beta sequence with a U-shape accounts for most finite sample properties, we preserve these two sequences and then pseudo-randomly generated ten sequences to form another 12 sequences. Their efficacy in approximating the finite sample structure of uniform random variables was evaluated by solving the above system of linear equations for the first four moments. Sequences that met the predetermined accuracy threshold (error less than $10^{-10}$) were retained, while those that did not meet the requirement were discarded in favor of a new set. Upon identifying enough qualified sets, these sets were applied to assess the finite sample biases or variances in other estimators for the Gaussian distribution. The outcomes indicate that using merely 10 sets of sequences, totaling 120 sequences, which can be executed on a standard PC in a negligible amount of time, can achieve a accuracy of approximately 0.005 (root mean square error, RMSE) for the finite sample bias of standard deviation. In contrast, the RMSE is about 0.012 if using classic Monte Carlo methods with the same size of sequences.

\begin{table*}
\centering
\caption{The root mean square error of using different kinds of sequenses to estimate the finite sample behaviour of standard deviation for the Gaussian distribution}
\begin{tabular}{lllllllllll}
\hline
Bias-sd-G & Random 10S & Arithmetic & 12 Designed 10S & BAR-G 10S & BAR-G 20S & BAR-G 30S & Random 50S & BAR-G 50S & BAR-5D 50S \\ \hline
RMSE & 0.0129         & 0.0736     & 0.0274              & 0.0051      & 0.005         & 0.0047        & 0.0044         & 0.0049        & 0.0024            \\ \hline  

\bottomrule
\end{tabular}
\begin{tabular}{lllllllll}
\hline
Variance-sd-G & Random 10S & 12 Designed 10S& BAR-G 10S & BAR-G 20S & BAR-G 30S  & Random 50S& BAR-G 50S & BAR-5D 50S \\ \hline
RMSE & 0.0085 & 0.1342      & 0.0350                           & 0.0350       & 0.0340    & 0.0032    & 0.0341       &0.0307  \\ \hline
\bottomrule
\end{tabular}
\begin{tabular}{lllllllllll}
\hline
Bias-median-E & Random 10S & Arithmetic & BAR-E 10S & BAR-E 20S & BAR-E 30S & Random 50S & BAR-E 50S & BAR-5D 50S \\ \hline
RMSE & 0.0195         & 0.0213     & 0.0112              & 0.0079      & 0.0073         & 0.0074        & 0.0070         & 0.0071             \\ \hline  

\bottomrule
\end{tabular}

\label{tab:comparison}
\begin{minipage}{1\linewidth}
\footnotesize S: set; sd: standard deviation; G: Gaussian distribution; E: Exopnential distribution; 5D: five two-parametric distributions.

\end{minipage}
\end{table*}

%


There are still limitations with these ten sets of betaU-arith-random (BAR) sequences. 1, the obtained results, although much better than the 12 well-designed sequences and Monte Carlo methods, are still constantly biased (Figure \ref{fig:3}). 2, the accuracy cannot be further improved by simply increasing the sets of sequences (Table \ref{tab:comparison}). Theorem \ref{whlk} indicates that if finding sets of sequences that are consistent for other distributions, as the kinds of distributions grow, the combined sets of sequences will approach true randomness. This suggests another way to further improve the accuracy. 

Here, we calibrate 30 sets of sequences for the Weibull ($k=2$, $k=5$), gamma ($k=1$), lognormal ($\sigma=0.25$, $\sigma=0.5$), Pareto ($\alpha=7$, $\alpha=10$, $\alpha=15$), and generalized Gaussian distributions ($\beta=2$, $\beta=4$), we found that, by combing these 360 sequences with 20 sets of sequences calibrated from the Gaussian distribution, the accuracy can be reach 0.0024 (RMSE) for the finite sample bias of standard deviation in the Gaussian distribution. In contrast, the RMSE is about 0.0049 in classic Monte Carlo methods with the same size of sequences (Table \ref{tab:comparison}). 

Besides the standard deviation, theoretically, this method can be used for any estimator, another example we provided here is the median for the exponential distribution. Besides the calibration to the Gaussian distribution changed to the exponential distribution, other procedures are the same. The results show similar performance as the standard deviation for the Gaussian distribution when the number of sequences is 120, indicating the universality of this approach (Figure \ref{fig:4}, Table \ref{tab:comparison}).

\begin{figure*}
\centering
\includegraphics[width=1\linewidth]{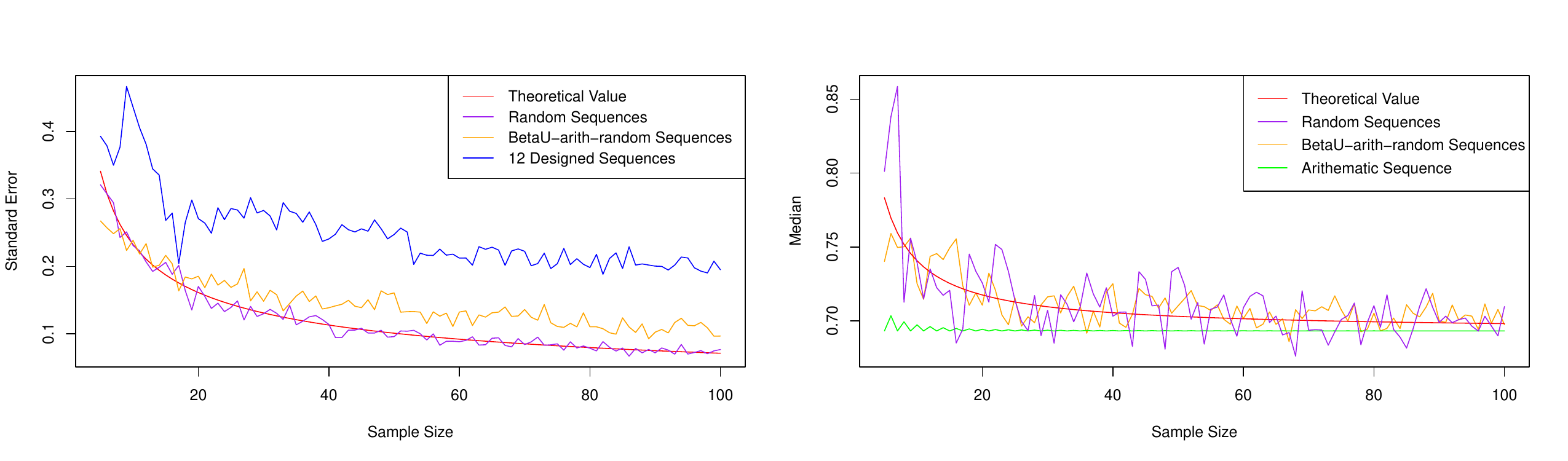}
\caption{The first plot shows the estimations of the standard errors of standard deviation for the Gaussian distribution by using different sequences as the sample size increases. The second plot depicts the sample median estimated from random, beta-arith-random, and arithmetic sequences and compares them to the true values. The numbers of random and betaU-arith-random sequences were both 120 in total.}

\label{fig:4}
\end{figure*}

\section*{Variance}\label{2}
The calibrated sequences can also be used to estimate the variance of sample standard deviation for the Gaussian distribution (Figure \ref{fig:4}, Table \ref{tab:comparison}). However, the accuracy is much lower than the estimation of finite sample bias and the Monte Carlo approach. When considering the variance of median for the exponential distribution, there is another concern. As one of the fundamental theorems in statistics, the central limit theorem declares that the standard deviation of the limiting form of the sampling distribution of the sample mean is $\frac{\sigma}{\sqrt n}$. The principle, asymptotic normality, was later applied to the sampling distributions of robust location estimators \cite{newcomb1886generalized, daniell1920observations,mosteller1946some,rao1952advanced,bickel1967some,chernoff1967asymptotic,lecam1970assumptions,bickel2012descriptive1,bickel2012descriptive2,janssen1984asymptotic}. Daniell (1920) stated \cite{daniell1920observations} that comparing the efficiencies of various kinds of estimators is useless unless they all tend to coincide asymptotically. However, the median and mean do not coincide asymptotically for the exponential distribution. Bickel and Lehmann, in 1976 \cite{bickel2012descriptive1,bickel2012descriptive2}, argued that meaningful comparisons of the efficiencies of various kinds of location estimators can be accomplished by studying their standardized variances, asymptotic variances, and efficiency bounds. Standardized variance, $\frac{\text{Var}\left(\hat{\theta}\right)}{\theta^2}$, allows the use of simulation studies or empirical data to compare the variances of estimators of distinct parameters. However, a limitation of this approach is the inverse square dependence of the standardized variance on $\theta$. If $\text{Var}\left(\hat{\theta}_1\right)=\text{Var}\left(\hat{\theta}_2\right)$, but $\theta_1$ is close to zero and $\theta_2$ is relatively large, their standardized variances will still differ dramatically. Here, the scaled standard error (SSE) is proposed as a method for estimating the variances of estimators measuring the same attribute, offering a standard error more comparable to that of the sample mean and much less influenced by the magnitude of $\theta$.

\begin{definition}[Scaled standard error]Let $\mathcal{M}_{s_i s_j}\in \mathbb{R}^{i\times j}$ denote the sample-by-statistics matrix, i.e., the first column corresponds to $\widehat{\theta}$, which is the mean or a $U$-central moment measuring the same attribute of the distribution as the other columns, the second to the $j$th column correspond to $j-1$ statistics required to scale, ${\widehat{\theta_{r_1}}}$, ${\widehat{\theta_{r_2}}}$, $\ldots$, ${\widehat{\theta_{r_{j-1}}}}$. Then, the scaling factor ${\mathcal{S}=\left[1,\frac{\bar{\theta_{r_1}}}{\bar{\theta_m}},\frac{\bar{\theta_{r_2}}}{\bar{\theta_m}},\ldots,\frac{\bar{\theta_{r_{j-1}}}}{\bar{\theta_m}}\right]}^T$ is a $j\times1$ matrix, which $\bar{\theta}$ is the mean of the column of $\mathcal{M}_{s_i s_j}$.  The normalized matrix is $\mathcal{M}_{s_i s_j}^N=\mathcal{M}_{s_i s_j}\mathcal{S}$. The SSEs are the unbiased standard deviations of the corresponding columns of $\mathcal{M}_{s_i s_j}^N$. \end{definition}

The $U$-central moment (the central moment estimated by using $U$-statistics) is essentially the mean of the central moment kernel distribution, so its standard error should be generally close to $\frac{\sigma_{\mathbf{k}m}}{\sqrt n}$, although not exactly since the kernel distribution is not i.i.d., where $\sigma_{\mathbf{k}m}$ is the asymptotic standard deviation of the central moment kernel distribution. If the statistics of interest coincide asymptotically, then the standard errors should still be used, e.g, for symmetric location estimators and odd ordinal central moments for the symmetric distributions, since the scaled standard error will be too sensitive to small changes when they are zero. The scaled standard error of median is higher than that of the mean for the exponential distribution, while their standard errors are nearly identical (Figure \ref{fig:5}).

\begin{figure}
\includegraphics[width=1\linewidth]{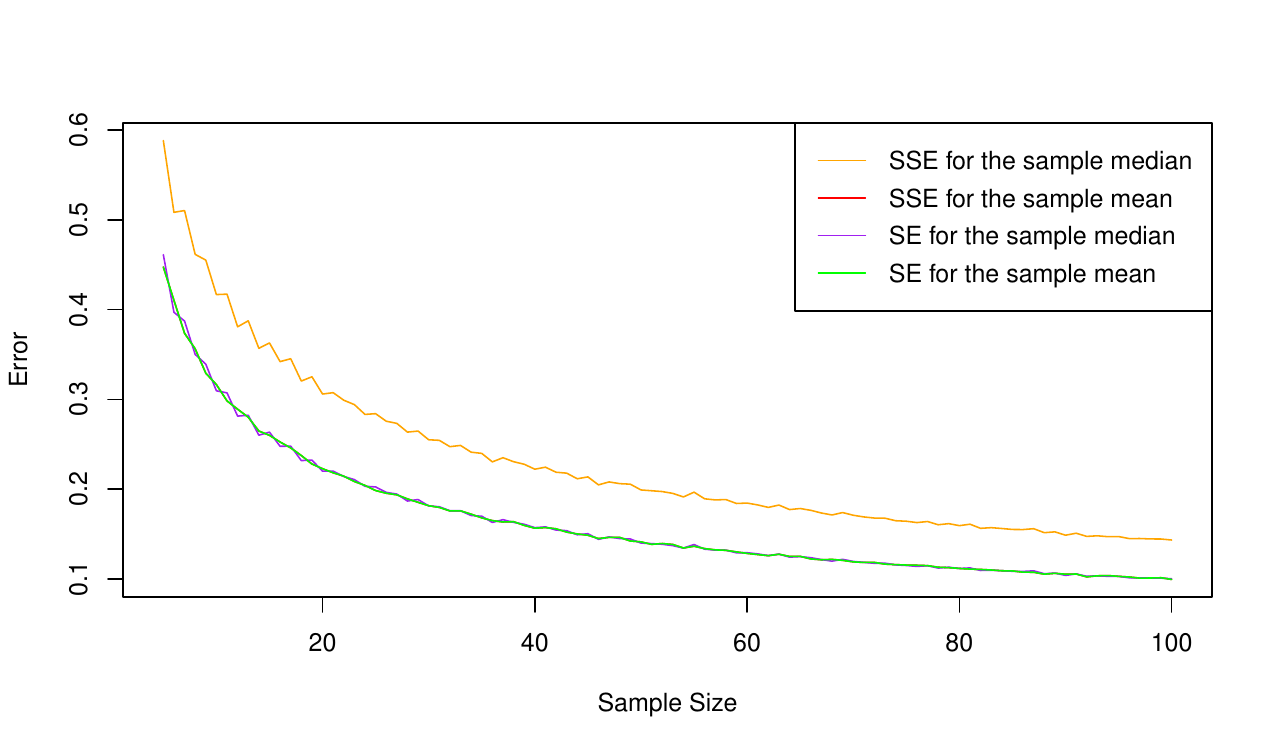}
\caption{The standard errors and scale standard errors of sample median and sample mean with respect to the sample size.}

\label{fig:5}
\end{figure}



\showmatmethods{} 
\section*{Data and Software Availability} All data are included in the main text. All codes have been deposited in \href {https://github.com/johon-lituobang/REDS_Nonasymptotic} {github.com/johon-lituobang/REDS}. 
\acknow{I sincerely acknowledge the insightful comments from Peter Bickel, which considerably elevating the lucidity and merit of this paper.}

\showacknow{} 

\bibliography{pnas-sample}

\end{document}